\documentclass[11pt,reqno]{amsart}

\usepackage{amsmath, euscript, amssymb, amsfonts}
\usepackage[matrix,arrow,curve]{xy}

\usepackage{mathrsfs}

\theoremstyle{plain}
\newtheorem{theorem}{Theorem}

\newtheorem{corollary}{Corollary}
\newtheorem{proposition}{Proposition}

\theoremstyle{definition}

\newtheorem*{acknowledgements}{Acknowledgements}

\theoremstyle{remark}
\newtheorem{remark}{Remark}

\renewcommand{\Re}{\mathop{\mathrm{Re}}\nolimits}
\renewcommand{\Im}{\mathop{\mathrm{Im}}\nolimits}

\newcommand{\mS}{{\mathscr{S}}}
\newcommand{\mW}{{\mathscr{W}}}
\newcommand{\mF}{{\mathscr{F}}}
\newcommand{\mB}{{\mathscr{B}}}
\newcommand{\mE}{{\mathscr{E}}}

\newcommand{\oR}{{\mathbb R}}
\newcommand{\oC}{{\mathbb C}}
\newcommand{\oZ}{{\mathbb Z}}

\newcommand{\labelT}[1]{\label{T#1}}
\newcommand{\labelP}[1]{\label{P#1}}

\newcommand{\twast}{\mathbin{\circledast}}

\renewcommand{\Re}{\mathop{\mathrm{Re}}\nolimits}
\renewcommand{\Im}{\mathop{\mathrm{Im}}\nolimits}
\newcommand{\eqdef}{\stackrel{\mathrm{def}}{=}}

\begin{document}

\title[Algebras with convergent star products and their
representations]{Algebras with convergent star products and their
representations in Hilbert spaces}
\author{M.~A.~Soloviev}
\address{I.~E.~Tamm Department of Theoretical Physics, Lebedev
Institute of Physics, Leninsky Prospect 53, 119991 Moscow, Russia}
\email{soloviev@lpi.ru}

\thanks{}


\begin{abstract}
We study  star product algebras of analytic functions for which
the power series defining the products converge absolutely. Such
algebras arise naturally in deformation quantization theory and in
noncommutative quantum field theory. We consider different star
products in a unifying way and  present  results on the structure
and basic properties of these algebras, which are useful for
applications. Special attention is given to the Hilbert space
representation of the algebras  and to the exact description  of
their corresponding operator algebras.
\end{abstract}

\hfill  FIAN-TD-2012-02

\medskip

\hfill J. Math. Phys. {\bf 54}, 073517 (2013)


\vspace{2cm}

\maketitle

\section{Introduction}\label{S1}

In this paper, we study  star product algebras of analytic
functions for which the power series defining the products are
absolutely convergent. We discuss the structure and some basic
properties of these algebras, consider their representations in
Hilbert spaces, and describe precisely the corresponding operator
algebras. The question of con\-vergence of star products was
previously considered in the context of deformation quantization
and specifically for the Weyl-Moyal and Wick
products~\cite{1}--\cite{4}. It acquires a new significance in
noncommutative quantum field theory (NC QFT), because taking the
nonlocal nature of the star products into account is important for
the physical interpretation of noncommutative
models~\cite{5}--\cite{8}. As shown in~\cite{6},  a causality
condition for NC QFT can be formulated through the use of  a star
product algebra associated with the light cone and having the
property of absolute convergence, and some simple models satisfy
this condition. Since in the literature there is no consensus on
the equivalence of the Weyl-Moyal and Wick-Voros formulations of
noncommutative field theory~\cite{9}--\cite{11}, it is useful  to
define and explore those properties of the convergent star product
algebras that are independent of the specific form of the star
product.

We consider different quantizations in a unifying way, starting
from  a general form of  star products defined by a constant
matrix  whose antisymmetric part is the Poisson tensor and whose
symmetric part specifies the quantization. In the case of the
Weyl-Moyal product  the symmetric part is zero. An important role
in our approach is played by function spaces that are algebras
under each of these products and consist of rapidly decreasing
functions. As well known, the usual axiomatic quantum field theory
uses the Schwartz space $S$ of all rapidly decreasing smooth
functions, on which quantum fields are defined as operator-valued
distributions. The space $S$ is an algebra under the Weyl-Moyal
product but not under the Wick product, for which the symmetric
part is pure imaginary diagonal. Moreover, any test function space
having the structure of an algebra under the Wick product must
consist of entire analytic functions. This suggests that any
attempt to develop a Wick-Voros formulation of noncommutative QFT
should be based on treating it  as an essentially  nonlocal
theory, which can also be confirmed by the simple but instructive
example of the normal ordered star product square of a free field,
considered for the Weyl-Moyal and Wick products in~\cite{8}
and~\cite{12}. We believe that using appropriate  function
algebras consisting of rapidly decreasing functions and universal
for different star products is very useful in the general theory
of quantization on symplectic linear spaces. In particular, their
operator representations associated with a representation of the
translation group are defined in an obvious way and can be
extended by continuity to larger algebras. In~\cite{3}, a faithful
representation of a Wick star product algebra with absolute
convergence is defined by means of the Gelfand-Naimark-Segal
construction. We construct Hilbert space representations of
various star product algebras explicitly and characterize the
domain of definition of the corresponding operator algebras.

The paper is organized as follows. In Sec.~\ref{S2}, we introduce
the main definitions and notation. In Sec.~\ref{S3}, we consider
the spaces $\mathcal E^\rho$ consisting of entire  functions of
order $\rho\le2$ and of minimal type, and prove that these spaces
are topological algebras under any star product $\star_\vartheta$
defined by a constant matrix $\vartheta$. In Sec.~\ref{S4}, we
show that the algebra $(\mathcal E^\rho,\star_\vartheta)$ contains
a family of subalgebras associated naturally with  cone-shaped
regions. If $\rho>1$, then $(\mathcal E^\rho,\star_\vartheta)$ has
a subalgebra $\mW^\rho$ consisting of functions rapidly decreasing
at real infinity.  These are just the functions that can serve as
test functions for NC QFT, and we explain the relation of
$\mW^\rho$ with the Gelfand-Shilov spaces $S^\beta$ used in the
previous papers~\cite{6}, \cite{8}, and \cite{12}. In
Sec.~\ref{S5}, we describe the Fourier transform of the algebra
$(\mW^\rho,\star_\vartheta)$, which is used later in Sec.~\ref{S6}
to construct Hilbert space representations of this algebra. In
Sec.~\ref{S7}, we prove that the representations of
$(\mW^\rho,\star_\vartheta)$ can be uniquely extended by
continuity to  representations of the algebra $(\mathcal
E^\rho,\star_\vartheta)$. We also characterize exactly the domain
of definition of the corresponding operator algebra in the space
$L^2$ of square  integrable functions. Analogous results for the
Fock-Bargmann space representation are derived in Sec.~\ref{S8}.
The final Sec.~\ref{S9} contains concluding remarks.

\section{Star products on symplectic linear spaces}\label{S2}

Let $X$ be  be a real $2n$-dimensional vector space equipped with
a  symplectic bilinear form $\omega$. Choosing a symplectic basis
in $X$, we  identify the form $\omega$ with its matrix
$\Omega=\begin{pmatrix} 0&I_n\\-I_n&0\end{pmatrix}$ and we denote
by $\pi$  the inverse matrix of $\Omega$. Then the Poisson bracket
of two functions on $X$ is written as
 \begin{equation}
\{f,g\}= \sum_{j,k=1}^{2n} \pi^{jk}\frac{\partial f}{\partial
x^j}\frac{\partial g}{\partial x^k}.
 \notag
\end{equation}
The Weyl-Moyal star product $f\star_{\hbar}g$ is a noncommutative
deformation of the ordinary (pointwise) product of the functions
in the direction of the Poisson bracket and is defined by
 \begin{equation}
(f\star_{\hbar}
g)(x)=f(x)\,e^{\tfrac{i\hbar}{2}\,\overleftarrow{\partial_j}\,
\pi^{jk}\,\overrightarrow{\partial_k}}g(x),
 \label{2.1}
\end{equation}
where the Planck constant  $\hbar$ plays the role of a
noncommutativity parameter, and where  the summation convention
for the repeated indices is used. In this paper, we consider star
products of the more general form
\begin{multline}
(f\star_{\hbar\vartheta}
g)(x)=f(x)\,e^{i\hbar\,\overleftarrow{\partial_j}\,
\vartheta^{jk}\,\overrightarrow{\partial_k}}g(x)=\\=
f(x)g(x)+\left.\sum_{m=1}^\infty\frac{(i\hbar)^m}{m!}\,
\left(\vartheta^{jk}\partial_{x^j}\,
\partial_{x'^k}\right)^m f(x)g(x')\right|_{\,x=x'},
 \label{2.2}
\end{multline}
where $\vartheta$ is  a constant and, in general, complex matrix.
By the classical-quantum corres\-pondence principle, the
antisymmetric part of $\vartheta$ must be equal to $\dfrac12 \pi$.
Then, and only then, the product~\eqref{2.2} satisfies the limit
relation
\begin{equation}
\lim_{\hbar\to0}\frac{1}{i\hbar}(f\star_{\hbar\,\vartheta}
g-g\star_{\hbar\,\vartheta} f)=\{f,g\}.
 \notag
\end{equation}
Therefore, in what follows we assume that
\begin{equation}
\vartheta=Q+\frac12 \pi,\qquad \pi=\Omega^{-1},\quad
Q^{jk}=Q^{kj}.
 \label{2.3}
\end{equation}
Different choices of the symmetric matrix $Q$ correspond to
different quantizations, i.e., to different rules of association
between functions on $X$ and quantum operators, including the
standard, the anti-standard, the normal, and the anti-normal
ordering rules. This is demonstrated below by constructing
explicitly  Hilbert space representations which convert the star
products to the operator product.  In  Sec.~\ref{S3}--\ref{S5} we
write $\star_\vartheta$ instead of $\star_{\hbar\vartheta}$,
assuming that the deformation parameter $\hbar$ is included in the
matrix $\vartheta$.

We note that the general form~\eqref{2.2} of star product was used
in recent papers by Blaszak and  Doma\'nski~\cite{13} and Omori et
al.~\cite{14}, but the attention of these papers was focused on
other aspects. The first of them presents the formalism of the
phase space quantum mechanics in a very systematic way, but does
not consider problems with convergence of the infinite series
representing the star products, whereas the second paper  focuses
mainly on the important case of the $\star$-exponential functions
of quadratic forms and their related Clifford algebra and does not
concern the representation of star product algebras in Hilbert
spaces.

\section{Convergence of the star products}\label{S3}

We will consider spaces of functions for which the star
product~\eqref{2.2} is absolutely convergent for any matrix
$\vartheta$. Such algebras were previously studied~\cite{1},
\cite{2} mostly for the case of Weyl-Moyal product. The
generalization to arbitrary star products in not difficult, but we
give it for completeness in the form convenient for deriving our
main results in Secs.~\ref{S4}--\ref{S8}.

Let $0<\rho<\infty$ and let   $\mathcal E^\rho(\oR^d)$ denote the
space of all smooth functions on $\oR^d$ satisfying the
inequalities
\begin{equation}
|\partial^\kappa f(x)|\leq C_L
L^{-|\kappa|}(\kappa!)^{1-1/\rho}\,e^{|x/L|^\rho},
 \label{3.1}
\end{equation}
where $L>0$ and can be chosen arbitrarily large,  $C_L$ is a
constant depending on $f$ and $L$,  $|x|=\max\limits_{1\le j\le d}
|x_j|$, and $\kappa=(\kappa_1,\dots,\kappa_d)\in \oZ_+^d$. Here
and hereafter, we use the standard multi-index notation
\begin{equation}
\partial^\kappa =\frac{\partial^{|\kappa|}}{\partial x_1^{\kappa_1}
\dots\partial x_d^{\kappa_d}},\quad
|\kappa|=\kappa_1+\dots+\kappa_d,\quad
\kappa!=\kappa_1!\dots\kappa_d!.
 \notag
\end{equation}
The topology of $\mathcal E^\rho(\oR^d)$ is defined by the family
of norms
\begin{equation}
\|f\|_L=\sup_{x,\kappa} \frac{L^{|\kappa|}|\partial^\kappa
f(x)|}{(\kappa!)^{1-1/\rho}}\,e^{-|x/L|^\rho},\quad L>0 .
 \label{3.2}
\end{equation}
We note that $\kappa!$ in~\eqref{3.1}, \eqref{3.2} can be replaced
with $\kappa^\kappa=\kappa_1^{\kappa_1}\dots\kappa_d^{\kappa_d}$,
because by  Stirling's formula we have
\begin{equation}
 e^{-|\kappa|}\kappa^\kappa\le \kappa!\le
  C_\delta(1+\delta)^{|\kappa|}e^{-|\kappa|} \kappa^\kappa
  \qquad\text{for any}\,\, \delta>0.
 \label{3.3}
\end{equation}
However the definition in form~\eqref{3.1}, \eqref{3.2} is more
convenient for what follows.

The definition of $\mathcal E^\rho(\oR^d)$ can be reformulated in
terms of complex variables. Let  $\mE^\rho(\oC^d)$ be the space of
all entire functions  of order $\le\rho$ and minimal type, i.e.
satisfying the condition
\begin{equation}
|f(z)|\leq C_\epsilon\, e^{\epsilon|z|^\rho},
 \label{3.4}
\end{equation}
where  $\epsilon>0$ and can be chosen arbitrarily small, $z=x+iy$,
and $|z|=\max\limits_{1\le j\le d} |z_j|$. The natural topology on
$\mE^\rho(\oC^d)$ is defined by the countable system of norms
\begin{equation}
\|f\|_\epsilon=\sup_z|f(z)| e^{-\epsilon|z|^\rho},\quad
\epsilon=1, \tfrac12,\tfrac13,\dots .
 \label{3.5}
\end{equation}
and is hence metrizable. It follows from the elementary properties
of analytic functions of several variables that $\mE^\rho(\oC^d)$
is complete and therefore a Fr\'echet space.

\begin{proposition}\labelP{1}  The space $\mE^\rho(\oC^d)$ is canonically
isomorphic to $\mathcal E^\rho(\oR^d)$.
\end{proposition}

\begin{proof}
The Taylor series expansion of $f\in \mathcal E^\rho(\oR^d)$ about
a point $x\in\oR^d$ has an infinite radius of convergence and
defines an analytic continuation of $f$ to $\oC^d$ which
satisfies~\eqref{3.4}. Indeed, using the first of
inequalities~\eqref{3.3} and choosing $L'<L$, we obtain
\begin{multline}
|f(x+iy)|\le\sum_\kappa\frac {|y^\kappa|}{\kappa!}|\partial^\kappa
f(x)|\leq \|f\|_L\, e^{|x/L|^\rho} \sum_\kappa\frac
{|y^\kappa|}{L^{|\kappa|}(\kappa!)^{1/\rho}}\le
\\\leq \|f\|_L\, e^{|x/L|^\rho} \sup_\kappa \frac {|y^\kappa|
e^{|\kappa|/\rho}}{{L'}^{|\kappa|} (\kappa^\kappa)^{1/\rho}}
\sum_\kappa\left(\frac{L'}{L}\right)^{|\kappa|}\leq
C_{L,L'}\|f\|_L\, e^{|x/L|^\rho+(d/\rho)|y/L'|^\rho}\leq\\\leq
C_{L,L'}\|f\|_L\, e^{(1+d/\rho)|z/L'|^\rho}. \label{3.6}
\end{multline}
Therefore, $f(x+iy)\in \mE^\rho(\oC^d)$ and $\| f\|_\epsilon\leq
C_{\epsilon,L}\|f\|_L$ if $(1+d/\rho)L^{-\rho}<\epsilon$, and thus
the corresponding map from  $\mathcal E^\rho(\oR^d)$ to
$\mE^\rho(\oC^d)$ is continuous.

Conversely, let $f\in \mE^\rho(\oC^d)$ and let $D_r=\{\zeta\in
\oC^d\colon |\zeta_j|\leq r_j, j=1,\dots,d\}$, where all $r_j$'s
are positive. It follows from the Cauchy inequality that
\begin{equation}
|\partial^\kappa f(z)|\leq \frac{\kappa!}{r^\kappa}\sup_{\zeta\in
D_r}|f(z-\zeta)|\leq \frac{\kappa!}{r^\kappa}\|f\|_\epsilon
e^{\epsilon|2z|^\rho+\epsilon|2r|^\rho},
 \label{3.7}
 \end{equation}
where $|r|^\rho=\max_jr_j^\rho\le \sum_j r_j^\rho$. A simple
computation gives
\begin{equation}
\inf_r\frac{e^{\sum_jr_j^\rho}}{r^k}=\frac{(\rho
e)^{|\kappa|/\rho}}{(\kappa^\kappa)^{1/\rho}}. \label{3.8}
\end{equation}
Combining~\eqref{3.7} and~\eqref{3.8}, we obtain
\begin{equation}
|\partial^\kappa f(z)|\leq  \|f\|_\epsilon\,
2^{|\kappa|}(\epsilon\rho e)^{|\kappa|/\rho}\frac{\kappa!}
{(\kappa^\kappa)^{1/\rho}}e^{\epsilon|2z|^\rho}. \label{3.9}
\end{equation}
Using the second of inequalities~\eqref{3.3}, this can also be
written as
\begin{equation}
|\partial^\kappa f(z)|\le C_{\epsilon'} \|f\|_\epsilon\,
2^{|\kappa|} (\epsilon'\rho)^{|\kappa|/\rho}(k!)^{1-1/\rho}
e^{\epsilon|2z|^\rho}\qquad \text{for any}\,\, \epsilon'>\epsilon.
 \label{3.10}
\end{equation}
We infer that the restriction $f|_{\oR^d}$ of the function to the
real space  belongs to $\mathcal E^\rho(\oR^d)$ and satisfies
\begin{equation}
\|f|_{\oR^d}\|_L\leq C'_{\epsilon,L}\|f\|_\epsilon\qquad \text{for
$L<\frac{1}{2\epsilon^{1/\rho}}\min\left(1,\frac{1}{\rho^{1/\rho}}\right)$\,},
 \notag
\end{equation}
which completes the proof.
\end{proof}

\begin{remark}\label{R1}  We prefer to say that $\mE^\rho(\oC^d)$ is
isomorphic to $\mathcal E^\rho(\oR^d)$ instead of saying that
these spaces coincide, because the elements of $\mathcal
E^\rho(\oR^d)$ can also be continued as anti-analytic functions.
\end{remark}

We now characterize the infinite order differential operators that
are endomorphisms of  $\mathcal E^\rho(\oR^d)$.

\begin{theorem}\labelT{1} Let $\rho>1$ and $\rho'=\rho/(\rho-1)$. If
$G(s)=\sum_\kappa c_\kappa s^\kappa$ is an entire function of
order $\le\rho'$ and  finite type, then the differential operator
$G(\partial)=\sum_\kappa c_\kappa \partial^\kappa$ maps $\mathcal
E^\rho(\oR^d)$ continuously into itself and the series
$\sum_\kappa c_\kappa \partial^\kappa f$ converges absolutely in
every norm of this space for any $f\in\mathcal E^\rho(\oR^d)$. If
$0<\rho\le1$, then analogous statements  hold for each  operator
of the form $G(\partial)$, where $G(s)$ is an entire function.
\end{theorem}

\begin{proof}
In deriving~\eqref{3.9}, we have already reproduced the well-known
estimate of the Taylor coefficients of entire functions with a
given order of growth. When applied to a function $G(s)$ of order
$\le \rho'$ and type $<a$, a similar estimate shows that
\begin{equation}
|c_\kappa|=\frac{|\partial^\kappa G(0)|}{\kappa!}\leq
C\frac{(a\rho')^{|\kappa|/\rho'}}{(\kappa!)^{1/\rho'}}\,,
 \label{3.11}
\end{equation}
where $C$ is a positive constant. If $1/\rho'=1-1/\rho$ and
$\epsilon''\ge 2^\rho\epsilon$,  then it follows from~\eqref{3.10}
and~\eqref{3.11} that
\begin{equation}
\sum_\kappa\|\,c_\kappa\partial^\kappa
f\|_{\epsilon''}\equiv\sum_\kappa |c_\kappa|\sup_z|\partial^\kappa
f(z)|e^{-\epsilon''|z|^\rho}\le C_{\epsilon'}\|f\|_\epsilon
\sum_\kappa 2^{|\kappa|}(\epsilon'
\rho)^{|\kappa|/\rho}(a\rho')^{|\kappa|/\rho'}, \notag
\end{equation}
where $\epsilon'>\epsilon$ and can be taken arbitrarily close to
$\epsilon$. Because $\epsilon$ can be chosen arbitrarily small, we
conclude that the series $\sum_\kappa c_\kappa
\partial^\kappa f$ converges absolutely in $\mE^\rho( \oC^d)\approx\mathcal
E^\rho(\oR^d)$ and the operator $G(\partial)$ is continuous in the
topology of this space. If $\rho\le1$, then it suffices to take
into account that the Taylor coefficients of any entire function
satisfy the inequality $|c_\kappa|\le
C_\varepsilon\varepsilon^{|\kappa|}$ with arbitrarily small
$\varepsilon>0$. Combining this inequality with~\eqref{3.10}, we
see that in this case, any operator of the form $G(\partial)$,
with $G(s)$ an entire function, maps $\mathcal E^\rho(\oR^d)$
continuously into itself. The theorem is proved.
\end{proof}

\begin{proposition}\labelP{2} The set of polynomials is dense in
$\mE^\rho(\oC^d)$.
\end{proposition}

\begin{proof}
 We show that the Taylor series expansion of $f\in
\mE^\rho(\oC^d)$ converges absolutely in $\mE^\rho(\oC^d)$. Let
$\epsilon'>2^\rho d\epsilon$. Using the inequality $d|z|^\rho\ge
\sum_j|z_j|^\rho$ together with~\eqref{3.8} and \eqref{3.9}, we
get
\begin{multline}
\sum_\kappa\left\|\frac {z^\kappa}{\kappa!}\partial^\kappa
f(0)\right\|_{\epsilon'}\le\sum_\kappa \frac {|\partial^\kappa
f(0)|}{\kappa!}\sup_z
|z^\kappa|e^{-(\epsilon'/d)\sum_j|z_j|^\rho}\le\\\le
\|f\|_\epsilon \sum_\kappa 2^{|\kappa|}(\epsilon
d/\epsilon')^{|\kappa|/\rho}
 \le C_{\epsilon,\epsilon'}\|f\|_\epsilon.
\notag
\end{multline}
Because the topology of $\mE^\rho(\oC^d)$ is finer than the
topology of simple convergence, we conclude that the Taylor series
converges absolutely in the topology of $\mE^\rho(\oC^d)$ to the
function $f$. The proposition is proved.
\end{proof}

\begin{theorem}\labelT{2} If $\rho\le 2$, then for any matrix $\vartheta$, the space
$\mathcal E^\rho(\oR^d)$ is an associative unital topological
algebra under the star product $f\star_\vartheta g$. The series
defining this product converges absolutely in $\mathcal
E^\rho(\oR^d)$.
\end{theorem}

\begin{proof}
 According to the definition~\eqref{2.2}, the product
$f\star_\vartheta g$ is obtained by applying the operator
$\exp(i\vartheta^{jk}\partial_j\otimes\partial_k)$ to the function
$(f\otimes g)(x,x')\in \mathcal E^\rho(\oR^{2d})$ and then
restricting to the diagonal  $x=x'$. The entire function
$\exp(i\vartheta^{jk}s_j s'_k)$ is of order 2 and finite type.
From the isomorphism $\mathcal E^\rho(\oR^d)\approx
\mE^\rho(\oC^d)$, we immediately infer that the restriction map
$\mathcal E^\rho(\oR^{2d})\to\mathcal E^\rho(\oR^d)$ is
continuous. Therefore Theorem~\ref{T1} implies that the series
defining $f\star_\vartheta g$ is absolutely convergent. The
bilinear map $(f,g)\to f\otimes g$ from $\mathcal
E^\rho(\oR^d)\times \mathcal E^\rho(\oR^d)$ to $\mathcal
E^\rho(\oR^{2d})$ is obviously jointly continuous. Therefore the
bilinear map $(f,g)\to f\star_\vartheta g$ is also jointly
continuous. To prove the associativity of this product we use the
identity
\begin{equation}
\partial_x^\kappa f(x,x)=\left.(\partial_x+\partial_{x'})^\kappa
f(x,x')\right|_{x=x'},
 \label{3.12}
\end{equation}
which holds for any $f\in \mathcal E^\rho(\oR^{2d})$. With the
help of~\eqref{3.12} we obtain
\begin{equation}
(f\star_\vartheta(g\star_\vartheta
h))(x)=\left.e^{i\vartheta^{jk}\partial_{x^j}(\partial_{x'^k}+
\partial_{x''^k})}
f(x)\left(e^{i\vartheta^{jk}\partial_{x'^j}\partial_{x''^k}}
g(x')h(x'')\right)\right|_{x=x'=x''} \label{3.13}
\end{equation}
\begin{equation}
((f\star_\vartheta g)\star_\vartheta h)(x)=\left.
e^{i\vartheta^{jk}(\partial_{x^j}+\partial_{x'^j})\partial_{x''^k}}
\left(e^{i\vartheta^{jk}\partial_{x^j}\partial_{x'^k}}f(x)g(x')\right)h(x'')\right|_{x=x'=x''}
\label{3.14}
\end{equation}
The right-hand sides of~\eqref{3.13} and~\eqref{3.14} coincide by
Theorem~\ref{T1} applied to the operator
$\exp({i\vartheta^{jk}\partial_{x^j}(\partial_{x'^k}+
\partial_{x''^k})+i\vartheta^{jk}\partial_{x'^j}\partial_{x''^k}})$
acting to $\mathcal E^\rho(\oR^{3d})$. Clearly, the function
$f(x)\equiv 1$ is the unit element of the algebra $(\mathcal
E^\rho(\oR^d),\star_\vartheta)$. The theorem is proved.
\end{proof}

\section{Subalgebras associated with cone-shaped regions}\label{S4}

If $\rho>1$, then the space $\mathcal E^\rho(\oR^d)$ has a
nontrivial subspace  $\mW^\rho(\oR^d)$ consisting of functions
rapidly decreasing at  real infinity.  More precisely,
$\mW^\rho(\oR^d)$ consists of smooth functions on $\oR^d$
satisfying the condition
\begin{equation}
|\partial^\kappa f(x)|\leq C_{N,L}
L^{-|\kappa|}(\kappa!)^{1-1/\rho}(1+|x|)^{-N}\quad\text{for
all}\,\, L>0,\,\,N=0,1,2,\dots.
 \label{4.1}
 \end{equation}
Correspondingly, its topology of  is defined by the system of
norms
\begin{equation}
\|f\|_{N,L}=\sup_{x,\kappa}L^{|\kappa|}(\kappa!)^{-(1-1/\rho)}(1+|x|)^N
|\partial^\kappa f(x)|,
 \label{4.2}
\end{equation}
and under this topology $\mW^\rho(\oR^d)$ is a Fr\'echet space. By
reasoning similar to the proof of Preposition~1, this space is
canonically isomorphic to the space of entire functions satisfying
the inequalities
\begin{equation}
(1+|x|)^N|f(x+iy)|\le C_{N,\epsilon}e^{\epsilon|y|^\rho},\qquad
\epsilon>0,\quad N=0,1,2,\dots,
 \label{4.3}
\end{equation}
and the system of norms~\eqref{4.2} is equivalent to the system
\begin{equation}
\|f\|_{N,\epsilon}=\sup_{x,y}(1+|x|)^N|f(x+iy)|e^{-\epsilon|y|^\rho}.
 \label{4.4}
\end{equation}
A similar space can be associated with each open cone
$V\subset\oR^d$. We define $\mW^\rho(V)$ to be the space of all
entire functions with the finite norms
\begin{equation}
\|f\|_{V,N,\epsilon}=\sup_{x,y}(1+|x|)^N|
f(x+iy)|e^{-\epsilon\delta_V^\rho(x)-\epsilon|y|^\rho},
 \label{4.5}
\end{equation}
where $\delta_V(x)=\inf_{\xi\in V}|x-\xi|$ is the distance of  $x$
from $V$. We note that $\delta_V(tx)=\delta_V(x)$ for all $t>0$,
because cones are invariant under dilations. Clearly, if
$V_1\subset V_2$, then $\mW^\rho(V_2)$ is canonically embedded in
$\mW^\rho(V_1)$. The space $\mW^\rho(\oR^d)$ is embedded in every
$\mW^\rho(V)$ and all of these spaces are embedded in $\mathcal
E^\rho(\oR^d)$. If $K$ is a closed cone and a functional $v\in
(\mW^\rho(\oR^d))'$ has a continuous extension to each
$\mW^\rho(V)$, where $V\supset K\setminus\{0\}$, then the cone $K$
can be thought of as a carrier of $v$, because this property
implies that  $v$ decreases  in all directions outside $K$ more
rapidly than any test function.

\begin{theorem}\label{T3} If $1\le\rho\le 2$, then for any open cone
$V\subset\oR^d$, the space $\mW^\rho(V)$ is an associative
topological algebra under the star product $f\star_\vartheta g$
and the series defining this product converges absolutely in
$\mW^\rho(V)$.
\end{theorem}

\begin{proof}
Let $f,g\in \mW^\rho(V)$. It is easy to see that $f\otimes g\in
\mW^\rho(V\times V)$ and
\begin{equation}
\|f\otimes g\|_{V\times V,N,2\epsilon}\le
\|f\|_{V,N,\epsilon}\|g\|_{V,N,\epsilon},
 \notag
\end{equation}
if $|y|$ in~\eqref{4.5} is defined by $|y|=\max_j|y_j|$. It is
clear that if $h(z,z')\in \mW^\rho(V\times V)$, then the function
$\check h(z)=h(z,z)$ belongs to $\mW^\rho(V)$ and $\|\check
h\|_{V,N,\epsilon}\le\|h\|_{V\times V,N,\epsilon}$. Therefore, it
suffices to prove an analogue of Theorem~\ref{T1} for
$\mW^\rho(V)$. To this end we derive an estimate similar
to~\eqref{3.10}. Using Cauchy's inequality and the elementary
inequality $(1+|x|)\le(1+|\xi|)(1+|x-\xi|)$, we obtain
\begin{multline}
(1+|x|)^N|\partial^\kappa f(z)|\leq
\frac{\kappa!}{r^\kappa}\sup_{\zeta=\xi+i\eta\in
D_r}(1+|x|)^N|f(z-\zeta)|\leq\\\leq
\frac{\kappa!}{r^\kappa}(1+|r|)^N\|f\|_{V,N,\epsilon}
e^{\epsilon\delta_V^\rho(2x)+\epsilon|2y|^\rho}\sup_{\zeta=\xi+i\eta\in
D_r} e^{\epsilon|2\xi|^\rho+\epsilon|2\eta|^\rho}\le\\\le
C_{N,\epsilon'}\frac{\kappa!}{r^\kappa}\|f\|_{V,N,\epsilon}
e^{\epsilon\delta_V^\rho(2x)+\epsilon|2y|^\rho
+2\epsilon'\sum_j(2r_j)^\rho}.
 \notag
 \end{multline}
Therefore,~\eqref{3.10} is  replaced by
\begin{equation}
(1+|x|)^N|\partial^\kappa f(z)|\leq
C_{N,\epsilon'}\|f\|_{V,N,\epsilon}\, 2^{|\kappa|}
(2\epsilon'\rho)^{|\kappa|/\rho}(k!)^{1-1/\rho}
e^{\epsilon\delta_V^\rho(2x)+\epsilon|2y|^\rho},
 \label{4.6}
\end{equation}
which holds for any $\epsilon'>\epsilon$. We now can state that if
$G(s)=\sum_\kappa c_\kappa s^\kappa$ is an entire function of
order $\le\rho'=\rho/(\rho-1)$ and finite type, then the
differential operator $G(\partial)=\sum_\kappa c_\kappa
\partial^\kappa$ maps $\mW^\rho(V)$ continuously into itself and
the series $\sum_\kappa c_\kappa
\partial^\kappa f$ converges absolutely in
this space. Indeed, if $\epsilon''>2^\rho\epsilon$, then it
follows from~\eqref{3.11} and~\eqref{4.6} that
\begin{multline}
\sum_\kappa\|\,c_\kappa\partial^\kappa
f\|_{V,N,\epsilon''}\equiv\sum_\kappa
|c_\kappa|\sup_{z=x+iy}(1+|x|)^N|\partial^\kappa
f(z)|e^{-\epsilon''\delta_V^\rho(x)-\epsilon''|y|^\rho}\\\le
C_{\epsilon'}\|f\|_{V,N,\epsilon} \sum_\kappa
2^{|\kappa|}(2\epsilon'
\rho)^{|\kappa|/\rho}(a\rho')^{|\kappa|/\rho'}, \notag
\end{multline}
where $\epsilon'>\epsilon$ and can be taken arbitrarily close to
$\epsilon$. In particular, the operator
$\exp(i\vartheta^{jk}\partial_j\otimes\partial_k)$ is a continuous
automorphism of $\mW^\rho(V\times V)$, and so the theorem is
proved.
\end{proof}

\begin{remark}\label{R2} The family of subalgebras $\mW^2(V)$ of the
maximal algebra $\mathcal E^2(\oR^d)$ with the absolute
convergence property plays a special role. In~\cite{6}
and~\cite{12}, we used the notation $\mS^{1/2}$ instead of
$\mW^2$, because this space can be regarded as the projective
limit of the Gelfand-Shilov spaces $S^{1/2,b}$ as $b\to0$. The
notation $\mW^2$ is here more convenient and  agrees with  the
notation in~\cite{15}, where the isomorphism between $S^{\beta,
b}$ and $W^{1/(1-\beta), b}$ is established for $\beta<1$. The
spaces $\mW^\rho=\bigcap_{b\to0}W^{\rho, b}$ should be
distinguished from the spaces $W^\rho=\bigcup_{b\to\infty}W^{\rho,
b}$ considered by Gelfand and Shilov~\cite{15}.  The space
$\mS^{1/2}=\mW^2$ was proposed in~\cite{6} as a universal test
function space for noncommutative quantum field theory. It was
argued there that the spaces $\mS^{1/2}(V)$ associated with cones
can be used to formulate a generalized causality condition and
shown that some simple noncommutative models  satisfy this
condition.
\end{remark}

\section{Deformed convolution product}\label{S5}

We now turn to describing the Fourier transform of the algebra
$(\mW^\rho,\star_\vartheta)$. As in the case~\cite{15} of the
spaces $W^\rho$, the behavior of the Fourier transforms of the
elements of  $\mW^\rho$ is characterized by an indicator function
which is the  Young dual of  $y^\rho$.

Let, as before,  $\rho>1$ and let $\rho'$  be the dual exponent of
$\rho$, i.e., $1/\rho'+1/\rho=1$. We define ${\mW}_{\rho'}(\oR^d)$
to be the space of all smooth functions on $\oR^d$ with the finite
norms
\begin{equation}
\|g\|_{N, a}=\max_{|\kappa|\le N}\sup_\sigma |\partial^\kappa
g(\sigma))|\,e^{a|\sigma|^{\rho'}}, \quad a>0,\,\,N=0,1,2,\dots.
 \label{5.1}
\end{equation}
In what follows we use the notation $\sigma\cdot x=\sum_j
\sigma_jx^j$, $\sigma\cdot \vartheta
t=\sum_j\sigma_j\vartheta^{jk}t_j$. The Fourier transform  is
defined by  $\hat f(\sigma)=(2\pi)^{-d/2}\int e^{-i\sigma\cdot
x}f(x)dx$.

\begin{theorem}\label{T4}  If $\rho'\ge2$, then for any complex matrix
$\vartheta$, the space ${\mW}_{\rho'}(\oR^d)$ is an associative
topological algebra under the  operation  $g_1,g_2\to
g_1\twast_\vartheta g_2$, where
\begin{multline}
(g_1\twast_\vartheta g_2)(\sigma)= \frac{1}{(2\pi)^{d/2}}\int
g_1(\sigma-\tau) g_2(\tau)e^{-i(\sigma-\tau)\cdot\vartheta \tau}
d\tau=\\\frac{1}{(2\pi)^{d/2}}\int g_1(\tau) g_2(\sigma-\tau)e^{-i
\tau\cdot\vartheta(\sigma-\tau)} d\tau.
 \label{5.2}
\end{multline}
The Fourier transformation is an isomorphism of $(\mW^\rho(\oR^d),
\star_\vartheta)$ onto $({\mW}_{\rho'}(\oR^d), \twast_\vartheta)$,
where $\rho'=\rho(\rho-1)$.
\end{theorem}

\begin{proof}
We note that the bilinear operation defined by~\eqref{5.2} is a
noncommutative deformation of the ordinary convolution (up to a
factor $(2\pi)^{-d/2}$). In the case when $\vartheta$ is a real
skewsymmetric matrix, it is  called the twisted convolution, and
this case corresponds to the Weyl-Moyal star product. For the
properties of the twisted convolution product and its relation to
the Weyl-Heisenberg group, we refer the reader to~\cite{16}
and~\cite{17}. Let us now show that the Laplace transformation
\begin{equation}
g\to f(x+iy)=(2\pi)^{-d/2}\int
e^{i\sigma\cdot(x+iy)}g(\sigma)d\sigma
 \label{5.3}
\end{equation}
maps ${\mW}_{\rho'}(\oR^d)$ onto $\mW^\rho(\oR^d)$. Under the
condition $\rho'\ge2$, the integral in~\eqref{5.3} is  well
defined and analytic in $\oC^d$, and we have the chain of
inequalities
\begin{multline}
(1+|x|)^N|f(x+iy)|\le C_N\max_{|\kappa\le N}|x^\kappa f(x+iy)|\le
C'_N\max_{|\kappa\le N}\int |e^{i\sigma\cdot(x+iy)}\partial^\kappa
g(\sigma) |d\sigma\\\le C'_N\|g\|_{N,(d+1)a'}\int e^{-\sigma\cdot
y-(d+1)a'|\sigma|^{\rho'}}d\sigma\le C_{N,a'}
\|g\|_{N,(d+1)a'}\sup_\sigma e^{-\sigma\cdot
y-a'\sum_j|\sigma_j|^{\rho'}}.
 \notag
\end{multline}
A standard calculation gives
\begin{equation}
\sup_\sigma\left(-\sigma\cdot
y-a'\sum_j|\sigma_j|^{\rho'}\right)=a\sum_j|y_j|^\rho\le
ad|y|^\rho,
 \notag
\end{equation}
where
\begin{equation}
(a'\rho')^\rho(a\rho)^{\rho'}=1.
 \label{5.4}
\end{equation}
We infer that $f\in \mW^\rho(\oR^d)$ and $\|f\|_{N,\epsilon}\le
C_{N,a'}\|g\|_{N,(d+1)a'}$ for $\epsilon\ge ad$. Conversely, let
$f$ belong to  $\mW^\rho(\oR^d)$. Then the integral
$(2\pi)^{-d/2}\int e^{-i\sigma\cdot(x+iy)}f(x+iy)\,dx$ is
independent of  $y$. Denoting it by $g(\sigma)$ and assuming that
$|\kappa|\le N$, we have
 \begin{multline}
|\partial^\kappa g(\sigma)|\le (2\pi)^{-d/2}\inf_y\int
 e^{\sigma\cdot y}\,|(x+iy)^\kappa f(x+iy)|dx\\
\le (2\pi)^{-d/2} \inf_y\int e^{\sigma\cdot y}(1+|y|)^N (1+|x|)^N |f(x+iy)|dx\\
\le C_{N,a}\|f\|_{N+d+1,\epsilon}\inf_y e^{\sigma\cdot
y+a\sum_j|y_j|^\rho}\le C_{N,a}\|f\|_{N+d+1,\epsilon}
e^{-a'|\sigma|^{\rho'}},
 \notag
\end{multline}
where $a>\epsilon$ can be taken arbitrarily close to $\epsilon$,
and $a'$ and $\rho'$ are defined by~\eqref{5.4}. We see that $g\in
{\mW}_{\rho'}(\oR^d)$ and the map $f\to g=\hat f$ from
${\mW}^\rho(\oR^d)$ to ${\mW}_{\rho'}(\oR^d)$ is continuous.

The Fourier transformation converts the  operator
$\exp(i\vartheta^{jk}\partial_j\otimes\partial_k)$ into the
multiplication by the function
$\exp\{-i\sigma_j\vartheta^{jk}\tau_k\}$ which is obviously a
pointwise multiplier of ${\mW}_{\rho'}(\oR^{2d})$ for $\rho'\ge
2$. Let $h(\sigma,\tau)$ be an arbitrary function in
${\mW}_{\rho'}(\oR^{2d})$. The formula
\begin{equation}
\frac{1}{(2\pi)^{d/2}}\int\left\{\frac{1}{(2\pi)^{d/2}}\int
h(\sigma-\tau,\tau)d\tau\right\}e^{i\sigma\cdot x}d\sigma
=\left.\frac{1}{(2\pi)^d}\int h(\sigma,\tau) e^{i\sigma\cdot
x+i\tau\cdot x'}d\sigma d\tau\right|_{x=x'}
 \notag
\end{equation}
shows that the  operation of restriction to the diagonal is
converted by the Fourier transformation into the operation
\begin{equation}
{\mW}_{\rho'}(\oR^{2d})\to {\mW}_{\rho'}(\oR^d)\colon
h(\sigma,\tau)\to \frac{1}{(2\pi)^{d/2}}\int
h(\sigma-\tau,\tau)d\tau.
 \notag
\end{equation}
Consequently,
\begin{equation}
\widehat{f_1\star_\vartheta f_2}=\hat f_1\twast_\vartheta\hat f_2,
 \label{5.6}
\end{equation}
which completes the proof.
\end{proof}

\section{Operator representations of star products}\label{S6}

Let  $T$ be a unitary projective representation of the translation
group of a symplectic space $(X,\omega)$ and let
$\exp\{\tfrac{i}{2\hbar}\omega\}$ be its multiplier, so that we
have the relation
\begin{equation}T_xT_{x'}=
e^{\tfrac{i}{2\hbar}\omega(x,x')}T_{x+x'}.
 \label{6.1}
 \end{equation}
We let $H$  denote  the Hilbert space on which this representation
acts. The Weyl quantization  associates Hilbert space operators
with functions on $X\approx \oR^{2n}$ in the following way
 \begin{equation}
f \longmapsto\mathfrak f=\frac{1}{(2\pi)^n}\int\! \tilde f(s)\,
T_{\hbar s} ds,
 \label{6.2}
\end{equation}
where $\tilde f$ is the symplectic Fourier transform of $f$,
defined by
\begin{equation}
\tilde f(s)=\frac{1}{(2\pi)^n}\int f(x)e^{i\omega(x,s)}dx.
 \label{6.3}
\end{equation}
Letting, as before,   $\Omega$ denote the matrix of the symplectic
form $\omega$, we have $\omega(x,s)=x\cdot \Omega s
=x^j\Omega_{jk}s^k$. Then $\tilde f(s)=\hat f(-\Omega s)$, and
after the change of variables $\sigma=-\Omega s$, the
formula~\eqref{5.6} can be written as
\begin{equation}
\widetilde{f_1\star_\vartheta f_2}=\tilde
f_1\twast_{\widetilde\vartheta}\tilde f_2,\qquad\text{where}\quad
\widetilde\vartheta=\Omega^t\vartheta\,\Omega=-\Omega\,\vartheta\,\Omega.
 \label{6.4}
\end{equation}
The integral in~\eqref{6.2} is well defined if $\tilde f$ is
integrable, and the Weyl transform takes the star
product~\eqref{2.1} into the operator product.

Let now the matrix $\vartheta$ be of the form given
by~\eqref{2.3}, and consequently,
\begin{equation}
\widetilde\vartheta=\widetilde Q-\frac12
\Omega,\qquad\text{where}\quad \widetilde Q=-\Omega Q\Omega.
 \label{6.5}
\end{equation}
Then the Weyl correspondence~\eqref{6.2} should be replaced by
\begin{equation}
 f \longmapsto\mathfrak f=\frac{1}{(2\pi)^n}\int\! \tilde
f(s)\,e^{\tfrac{i\hbar}{2} s\cdot \widetilde Qs} T_{\hbar s} ds.
 \label{6.6}
\end{equation}

\begin{theorem}\label{T5} If $1\le\rho\le2$, then for any $\vartheta$
of the form~\eqref{2.3}, the map~\eqref{6.6} is a continuous
homomorphism  of the algebra
$(\mW^\rho(\oR^{2n}),\star_{\hbar\vartheta})$ into the algebra
$\mathcal B(H)$ of bounded operators on the Hilbert space  $H$.
\end{theorem}

\begin{proof}
It follows from Theorem~\ref{T4} and from the invariance of
$\mW^\rho(\oR^{2n})$ under linear change of variables that the
symplectic Fourier transformation is a topological and algebraic
isomorphism of  $(\mW^\rho(\oR^{2n}),\star_{\hbar\vartheta})$ onto
 $(\mW_{\rho'}(\oR^{2n}),\twast_{\hbar\widetilde\vartheta})$.
Because $\rho'\ge2$, we have
\begin{equation}
\|\mathfrak f\|\le\frac{1}{(2\pi)^n}\int|\tilde
f(s)|e^{\tfrac{\hbar}{2}|\widetilde Q|\cdot|s|^2}ds\le C_a\|\tilde
f\|_{0,a},
 \notag
\end{equation}
where $a>(\hbar/2)|\widetilde Q|$. It remains to show that the map
 \begin{equation}
\mW_{\rho'}(\oR^{2n})\longrightarrow \mathcal B(H)\colon \tilde f
\longmapsto\mathfrak f=\frac{1}{(2\pi)^n}\int\! \tilde
f(s)\,e^{\tfrac{i\hbar}{2} s\cdot \widetilde Qs} T_{\hbar s} ds
\notag
\end{equation}
is an algebra homomorphism. Using~\eqref{6.1}, \eqref{6.5}, the
symmetry of the matrix $\widetilde Q$, and the Fubini theorem, we
obtain
\begin{multline}
\mathfrak f=\frac{1}{(2\pi)^n}\int\! (\tilde
f_1\twast_{\hbar\widetilde \vartheta}f_2)(s)\,e^{\tfrac{i\hbar}{2}
s\cdot \widetilde Qs} T_{\hbar s}ds=\\=\frac{1}{(2\pi)^{2n}}\iint
\tilde f_1(t)\tilde f_2(s-t)e^{-i\hbar\,
t\cdot\widetilde\vartheta(s-t) +\tfrac{i\hbar}{2} s\cdot
\widetilde Qs}
T_{\hbar s} dt ds =\\
=\frac{1}{(2\pi)^{2n}}\iint \tilde f_1(t)\tilde
f_2(t')e^{-i\hbar\,t\cdot\widetilde\vartheta t' +\tfrac{i\hbar}{2}
(t+t')\widetilde Q(t+t')} T_{\hbar(t+t') }dt
dt'=\\
=\frac{1}{(2\pi)^{2n}}\iint \tilde f_1(t)\tilde f_2(t')
e^{\tfrac{i\hbar}{2}t\cdot\widetilde
Qt+\tfrac{i\hbar}{2}t'\cdot\widetilde Qt'} T_{\hbar t}T_{\hbar
t'}dt dt'=\mathfrak f_1\mathfrak f_2.
 \notag
\end{multline}
The theorem is proved.
\end{proof}

In quantum mechanics on phase space, the  following notation is
usually used:  $x=(p_1,\dots,p_n; q^1,\dots, q^n)$, where $q^j$
are the coordinate variables and $p_j$ are their conjugate
momentums. The standard projective representation of the phase
space translations in the Hilbert space $L^2(\oR^n)$ is written as
\begin{equation}
T_{p,q}=e^{\tfrac{i}{\hbar}(p\mathfrak q-q\mathfrak p)},
 \label{6.8}
 \end{equation}
where  $\mathfrak q^j$ is  the operator of multiplication by the
$j$th coordinate function and  $\frac{i}{\hbar}\mathfrak p_j$ is
the operator of differentiation with respect to the same
coordinate. Let  $s=(u,v)$ and $\widetilde
Q=\dfrac12\left(\begin{matrix} 0&I_n\\I_n&0\end{matrix}\right)$.
Then  $s\cdot\widetilde Qs=u\cdot v$ and the
correspondence~\eqref{6.6}  takes the form
\begin{equation}
f \longmapsto\mathfrak f=\frac{1}{(2\pi)^n}\iint\! \tilde f(u,v)
e^{i\frac{\hbar}{2}u\cdot v+ i(u\mathfrak q-v\mathfrak p)} du dv=
\frac{1}{(2\pi)^n}\iint\! \tilde f(u,v) e^{iu\mathfrak
q}e^{-iv\mathfrak p} du dv,
 \label{6.9}
 \end{equation}
which defines the so-called $qp$-quantization, i.e., the standard
ordering of the operators. We let $j_{\rm st}$ denote  the map
defined by~\eqref{6.9}. In the theory of pseudodifferential
operators, it is usually called the Kohn-Nirenberg correspondence.
Under this correspondence the monomial $p^{\kappa} q^{\lambda}$
transforms into the operator $\mathfrak q^{\lambda}\mathfrak
p^{\kappa}$. Formally, this follows from the relations
\begin{equation}
\widetilde{p^{\kappa} q^{\lambda}}=(2\pi)^n\,
(-i)^{|\kappa|}i^{\lambda|}\partial^{\kappa}\delta(v)
\partial^{\lambda}\delta(u),\qquad
\partial^{\kappa}_v\partial^{\lambda}_u e^{iu\mathfrak
q}e^{-iv\mathfrak p}=(-i)^{|\kappa|}i^{\lambda|}\mathfrak
q^{\lambda} e^{iu\mathfrak q}\mathfrak p^{\kappa} e^{iv\mathfrak
p} , \notag
\end{equation}
but a rigorous proof is based on the extension  of map~\eqref{6.9}
to tempered distributions~\cite{18} An alternative proof, based on
extending it to the space $\mathcal E^2$, is given in
Sec.~\ref{S7}. The star product corresponding  to the standard
ordering is written as
\begin{equation}
(f\star_{\rm st}
g)(p,q)=f(p,q)\,e^{-i\hbar\,\overleftarrow{\partial_{p_j}}\,
\,\overrightarrow{\partial_{q^j}}}g(p,q).
 \label{6.10}
 \end{equation}
Let us point out  some other important cases. The matrix
$\widetilde Q=-\dfrac12\left(\begin{matrix}
0&I_n\\I_n&0\end{matrix}\right)$ defines the correspondence
\begin{equation}
f \longmapsto\mathfrak f=\frac{1}{(2\pi)^n}\iint\! \tilde f(u,v)
e^{-iv\mathfrak p} e^{iu\mathfrak q} du dv,
 \notag
 \end{equation}
under which the monomial $p^{\kappa} q^{\lambda}$  transforms to
$\mathfrak p^{\kappa}\mathfrak q^{\lambda}$, and so we have the
anti-standard ordering. Now, let $\widetilde
Q=-\dfrac{i}{2}\left(\begin{matrix}
\varsigma^{-2}I_n&0\\0&\varsigma^2I_n\end{matrix}\right)$, where
the dimensional parameter $\varsigma$ is such that $\varsigma q$
and $\varsigma^{-1}p$ are of the same dimension. In this case, it
is convenient to use the holomorphic variables
\begin{equation}
z^j=\frac{1}{\sqrt2}(\varsigma q^j+i\varsigma^{-1}p_j),\quad \bar
z_j=\frac{1}{\sqrt2}(\varsigma q^j-i\varsigma^{-1}p_j).
\label{6.11}
\end{equation}
By~\ref{P1}, the elements of $\mathcal E^\rho(\oR^{2n})$ and
$\mW^\rho(\oR^{2n})$ can be regarded as entire functions of these
variables, considered as independent variables. The restriction to
the real space $\oR^{2n}$ amounts then  to identifying $\bar z$
with the complex conjugate of $z$. With these variables, the
formula~\eqref{6.3} becomes
\begin{equation}
\tilde  f(w,\bar w)=\frac{1}{(2\pi i)^n}\iint f(z,\bar
z)e^{z\cdot\bar w-\bar z\cdot w}d z\wedge d\bar z,
 \notag
\end{equation}
where $w=(\varsigma u-i\varsigma^{-1}v)/\sqrt2$. As usual, we
introduce the annihilation and creation operators
\begin{equation}
 \mathfrak a^j=\frac{1}{\sqrt2}(\varsigma\mathfrak
q^j+i\varsigma^{-1}\mathfrak p_j),\qquad \mathfrak
a_j^\dagger=\frac{1}{\sqrt2}(\varsigma\mathfrak
q^j-i\varsigma^{-1}\mathfrak p_j).
 \notag
\end{equation}
Then the representation operator~\eqref{6.8} can be written as
\begin{equation}
T_z=e^{\tfrac{1}{\hbar}(z\mathfrak a^\dagger-\bar z \mathfrak
a)}=e^{-\tfrac{1}{2\hbar}z\bar z}e^{\tfrac{1}{\hbar}z\mathfrak
a^\dagger} e^{-\tfrac{1}{\hbar}\bar z\mathfrak a}.
 \notag
\end{equation}
In terms of the new variables, the quadratic form
$s\cdot\widetilde Qs$ becomes $-iw\cdot \bar w$, and the
correspondence~\eqref{6.6} takes the form
\begin{equation}
f \longmapsto \mathfrak f=\frac{1}{(2\pi i)^n}\iint \tilde
f(w,\bar w)\, e^{w\mathfrak a^\dagger} e^{-\bar w\mathfrak
a}\,dw\wedge d\bar w.
 \label{6.12}
\end{equation}
Therefore, in this case we have the Wick quantization, or in other
words, the normal ordering, under which the monomial $\bar
z^{\kappa} z^{\lambda}$   corresponds to the operator $(\mathfrak
a^\dagger)^{\kappa}\mathfrak a^{\lambda}$. The corresponding star
product is
\begin{equation}
(f\star_{Wick} g)(z,\bar z)=f(z,\bar z)
e^{\hbar\,\overleftarrow{\partial_z}\,
\cdot\overrightarrow{\partial_{\bar z}}}g(z,\bar z).
 \notag
\end{equation}
By changing the sign of $\widetilde Q$, we obtain the quadratic
form $iw\cdot \bar w$, which defines the anti-Wick quantization
and the anti-normal ordering  $\mathfrak a^{\lambda}(\mathfrak
a^\dagger)^{\kappa}$.

\section{Extensions of representations}\label{S7}
We now show that the map~\eqref{6.6} can be naturally extended to
functions  in $\mathcal E^\rho(\oR^{2n})$, where $\rho\le2$. For
clarity we first consider the representation in the space of
square integrable functions $L^2(\oR^n)$. In doing so we use
another function space whose elements  decrease at infinity faster
than the elements of $\mathcal E^\rho$ increase.

Let $\gamma>1$ and $\rho>1$. The space $W^\gamma_\rho(\oR^n)$
consists of all entire analytic functions  satisfying the
condition
\begin{equation}
|\psi(q+iy)|\leq C e^{-a|q|^\rho+b|y|^\gamma}
 \label{8.1}
 \end{equation}
where $a$, $b$, and $C$ are positive constants depending on
$\psi$. As before, we  really deal with the restrictions of these
functions to  $\oR^n$, using the analytic continuation as an
auxiliary tool. The space $W^\gamma_\rho$ coincides with the
Gelfand-Shilov space $S_{1/\rho}^{1-1/\gamma}$ and is nontrivial
if and only if $\gamma\ge\rho$. Its natural topology is the
inductive limit topology with respect to the family of Banach
spaces $W^{\gamma, b}_{\rho, a}$ equipped with the norms
\begin{equation}
\|\psi\|_{a,b}=\sup_{q,y}|\psi(q+iy)|e^{a|q|^\rho-b|y|^\gamma}.
 \label{8.2}
\end{equation}
Using reasoning similar to that in the proof of Theorem~\ref{T4},
it is easy to show that $W^\gamma_\rho$ is invariant under the
Fourier transformation if and only if $1/\gamma+1/\rho=1$. It is
important that $W^\gamma_\rho$ is dense in $L^2$. Indeed, if $f\in
L^2$, $\varphi\in W^\gamma_\rho$ and $\int \varphi(q)dq=1$, then
the sequence $f_\nu(q)=\nu^n\int f(q-\xi)\varphi(\nu\xi)d\xi$
belongs to $W^\gamma_\rho$ and converges to $f$ in the norm of
$L^2$ as $\nu \to\infty$.

We let $\mathcal A_{W^\gamma_\rho}(L^2)$ denote the algebra of
operators which are defined on the subspace $W_\rho^\gamma$ of the
Hilbert space $L^2$ and map $W_\rho^\gamma$ continuously into
itself, and we endow this algebra with the topology  of uniform
convergence on bounded subsets of $W_\rho^\gamma$.

\begin{theorem}\label{T6}  If $\,1<\rho\le2$, then the map~\eqref{6.6}
defined in Theorem~\ref{T5} as an algebra homomorphism  from
$\mW^\rho(\oR^{2n})$ to $\mathcal B(L^2)$ extends uniquely to a
continuous homomorphism of the algebra $(\mathcal
E^\rho(\oR^{2n}), \star_{\hbar\vartheta})$ to the algebra
$\mathcal A_{W^{\rho'}_\rho}(L^2)$, where $\rho'=\rho/(\rho-1)$.
\end{theorem}

\begin{proof}
Let $f(p,q)\in \mathcal E^\rho(\oR^{2n})$. This function can be
written as
 \begin{equation}
f(p,q)=\sum_{\kappa, \lambda}c_{\kappa,\lambda}
p^{\kappa}q^{\lambda},
 \label{8.5}
\end{equation}
where by the estimate~\eqref{3.9} the coefficients
$c_{\kappa,\lambda}=\dfrac{1}{\kappa!\lambda!}
\partial_p^{\kappa}\partial_q^{\lambda} f(0,0)$
satisfy
\begin{equation}
|c_{\kappa,\lambda}|\le \|f\|_\epsilon\, \frac{(2^\rho\epsilon
\rho e)^{(|\kappa|+|\lambda|)/\rho}}{(\kappa^\kappa
\lambda^\lambda)^{1/\rho}}.
 \label{8.6}
\end{equation}
We first extend the transformation~\eqref{6.9} that defines  the
standard ordering. If such a continuous  extension to $\mathcal
E^\rho$ exists, then it is unique because $\mW^\rho$ is dense in
$\mathcal E^\rho$ by~\ref{P2}. Indeed, every polynomial $P(z)$ can
be approximated in the topology of $\mE^\rho$ by a sequence of
functions of the form $f_\nu(z)=P(z)f_0(z/\nu)$, where $f_0\in
\mW^\rho$ is such that $f_0(0)=1$. Let $\psi\in
W^{\rho'}_\rho(\oR^n)$. We define an operator $\mathfrak f$ by
 \begin{equation}
 (\mathfrak f\psi)(q)\eqdef \sum_{\kappa, \lambda}c_{\kappa,\lambda}
(-i\hbar)^{|\kappa|} q^{\lambda}
\partial^{\kappa}\psi(q)
 \label{8.7}
 \end{equation}
and show that the series on the right-hand side of~\eqref{8.7}
converges to an element of $W^{\rho'}_\rho$, and $\mathfrak f$ is
hence well defined as an operator from $W^{\rho'}_\rho$ ito
itself. To this end, we estimate the derivatives of $\psi$ in the
complex domain, using  the Cauchy inequality again. Let $z=q+iy$.
Taking into account that $\rho'\ge\rho$, we get
\begin{multline}
|\partial^\kappa \psi(z)|\leq
\frac{\kappa!}{r^\kappa}\sup_{\zeta=\xi+i\eta\in
D_r}|\psi(z-\zeta)|\leq \frac{\kappa!}{r^\kappa}\|\psi\|_{a,b}
e^{-a|q/2|^\rho+b|2y|^{\rho'}}\sup_{\xi+i\eta\in D_r}
e^{a|\xi|^\rho+b|2\eta|^{\rho'}}\le\\\le
\frac{\kappa!}{r^\kappa}\|\psi\|_{a,b}
e^{-a|q/2|^\rho+b|2y|^{\rho'}}e^{a+(a+2^{\rho'}
b)\sum_jr_j^{\rho'}}.
 \notag
 \end{multline}
Using~\eqref{3.8} with   $\rho'$ instead of $\rho$, we obtain
 \begin{equation}
|\partial^\kappa\psi(z)|\le C\|\psi\|_{a,b} b_1^{|\kappa|/\rho'}
\frac{\kappa!}{(\kappa^\kappa)^{1/\rho'}}
e^{-a|q/2|^\rho+b|2y|^{\rho'}},\qquad\text{where}\quad
b_1=(a+2^{\rho'}b)\rho' e.
 \label{8.8}
\end{equation}
It follows from~\eqref{8.6} and~\eqref{8.8} that
\begin{multline}
\sum_{\kappa, \lambda}|c_{\kappa,\lambda}\hbar^{|\kappa|}
q^{\lambda}\partial^{\kappa}\psi(z)|\leq\\\leq
C\|f\|_\epsilon\|\psi\|_{a,b} e^{-a|q/2|^\rho+b|2y|^{\rho'}}
\sum_{\kappa} \frac{(2^\rho\epsilon \rho
e\hbar^\rho)^{|\kappa|/\rho}
b_1^{|\kappa|/\rho'}\kappa!}{(\kappa^\kappa)^{1/\rho}(\kappa^\kappa)^{1/\rho'}}
\sum_{\lambda} |q^{\lambda}|\frac{(2^\rho\epsilon \rho
e)^{|\lambda|/\rho}}{(\lambda^\lambda)^{1/\rho}},
 \notag
 \end{multline}
where
$(\kappa^\kappa)^{1/\rho}(\kappa^\kappa)^{1/\rho'}=\kappa^{\kappa}$.
Using~\eqref{3.3}, we see that the sum over $\kappa$ does not
exceed $\sum_{\kappa}(\epsilon \rho
eb_1\hbar^\rho)^{|\kappa|/\rho}$ and  is dominated by a constant
if $\epsilon<1/(\rho eb_1\hbar^\rho)$. The sum over $\lambda$ is
estimated similar to the sum in~\eqref{3.6} and does not exceed
$C_{\epsilon'}e^{\epsilon' n|2q|^\rho}$, where
$\epsilon'>\epsilon$ and can be taken arbitrarily close to
$\epsilon$. We infer that
\begin{equation}
\sum_{\kappa, \lambda}|c_{\kappa,\lambda}\hbar^{|\kappa|}
q^{\lambda}\partial^{\kappa}\psi(z)|\leq
C_\epsilon\|f\|_\epsilon\|\psi\|_{a,b}
e^{-(a/2)|q/2|^\rho+b|2y|^{\rho'}}
 \label{8.10}
\end{equation}
if $\epsilon$ is sufficiently small compared to  $a$ and $b$.
Therefore, $\mathfrak f \psi\in W^{\rho'}_\rho(\oR^n)$. The
operator $\mathfrak f$ maps  $W^{\rho'}_\rho$ into itself
continuously, because the right-hand side of~\eqref{8.10} contains
the factor $\|\psi\|_{a,b}$. Moreover, the Kohn-Nirenberg map
$j_{\rm st}\colon f\to\mathfrak f$ is continuous from $\mathcal
E^\rho$ to $\mathcal A_{W^{\rho'}_\rho}(L^2)$, because the
right-hand side of~\eqref{8.10} contains   $\|f\|_\epsilon$.

We now show that if $f\in \mW^\rho(\oR^{2n})$, then the definition
of  $\mathfrak f$ by~\eqref{8.7} is equivalent to the above
definition~\eqref{6.9}. To accomplish this, we consider the matrix
element $\langle\varphi,\mathfrak f\,\psi\rangle$, where $\varphi,
\psi\in W^{\rho'}_\rho(\oR^n)$ and the angle brackets denote the
inner product of $L^2(\oR^n)$. The function
\begin{equation}
\chi(u,v)=\langle \varphi,e^{iu\mathfrak q}e^{-iv\mathfrak
p}\psi\rangle=\int\bar\varphi(q)e^{iuq}\psi(q-\hbar v)dq
 \notag
\end{equation}
belongs to $W^{\rho'}_\rho(\oR^{2n})$, because
$\bar\varphi\otimes\psi\in W^{\rho'}_\rho(\oR^{2n})$ and this
space is invariant under  linear changes of variables and under
the partial Fourier transform. From~\eqref{6.9}, we have
\begin{multline}
\langle\varphi,\mathfrak f\,\psi\rangle=\frac{1}{(2\pi)^n}\iint
\tilde f(u,v)\chi(u,v)dudv=\frac{1}{(2\pi)^n}\iint
f(p,q)\tilde\chi(-p,-q) dpdq=\\= \frac{1}{(2\pi)^n}\sum_{\kappa,
\lambda}c_{\kappa\lambda}\iint p^{\kappa}
q^{\lambda}\tilde\chi(-p,-q) dpdq,
 \notag
\end{multline}
where the order of summation and integration can be interchanged
because of absolute convergence. Taking the (inverse) symplectic
Fourier transform, we obtain
 \begin{equation}
\langle\varphi,\mathfrak f\,\psi\rangle=\sum_{\kappa,
\lambda}c_{\kappa\lambda}
(-i)^{|\kappa|}i^{\lambda|}\left(\partial^{\kappa}_v
\partial^{\lambda}_u\delta, \chi\right)=
\sum_{\kappa, \lambda}c_{\kappa\lambda}
(-i\hbar)^{|\kappa|}\int\bar \varphi(q) q^{\lambda}
\partial^{\kappa}\psi(q)dq.
 \notag
\end{equation}
Thus, the  definitions~\eqref{6.9} and~\eqref{8.7} are indeed
consistent. The star product~\eqref{6.10} is continuous in
$\mathcal E^\rho$ by Theorem~\ref{T2}, and the operator product is
separately continuous in $\mathcal A_{W^{\rho'}_\rho}(L^2)$ by the
definition of the topology of bounded convergence. Because
$\mW^\rho$  is dense in $\mathcal E^\rho$, we conclude that the
constructed extension  of  correspondence~\eqref{6.9} is an
algebra homomorphism from $(\mathcal E^\rho,
\star_{\hbar\vartheta})$ to $\mathcal A_{W^{\rho'}_\rho}(L^2)$.

We let $Q_{\rm st}$ denote the matrix
$-\dfrac12\left(\begin{matrix} 0&I_n\\I_n&0\end{matrix}\right)$
defining the map $j_{\rm st}$ and let $j_Q$ denote the
symbol-operator correspondence of general form~\eqref{6.6}. If
$f\in \mW^\rho$, then for any $Q$, $\tilde f(s)e^{s\cdot\tilde
Qs}$ is the symplectic Fourier transform of
$e^{-Q^{jk}\partial_j\partial_k}f$. Hence we have the relation
\begin{equation}
j_Q(f)=j_{\rm st} \left(e^{\tfrac{i\hbar}{2}(Q^{jk}_{\rm
st}-Q^{jk})
 \partial_j\partial_k}f\right),
  \label{8.11}
 \end{equation}
which holds  for all $f\in \mW^\rho$.  We can now extend $j_Q$ to
a map from $\mathcal E^\rho$ into $\mathcal
A_{W^{\rho'}_\rho}(L^2)$ by replacing  $j_{\rm st}$
in~\eqref{8.11} with its extension constructed above. The map
defined in this way is continuous, because the differential
operator on the right-hand side of~\eqref{8.11} is a continuous
automorphism of $\mathcal E^\rho$ by Theorem~\ref{T1}. This
completes the proof.
\end{proof}

\section{The Fock-Bargmann space representation}\label{S8}

The Wick correspondence is usually realized  using a
representation in the Fock-Bargmann space $\mF^2$, which consists
of antiholomorphic functions on $\oC^n$ satisfying the condition
\begin{equation}
\int|\varphi(\bar z)|^2e^{-\frac{1}{\hbar}\|z\|^2}d\mu(z)<\infty,
 \notag
\end{equation}
where $d\mu(z) =(2\pi i)^{-n}dz\wedge \bar z=\pi^{-n}d(\Re z) d
(\Im z)$ and $\|z\|^2=z\cdot\bar z=\sum_j|z_j|^2$. It is a Hilbert
space with inner product given by
\begin{equation}
\langle\varphi,\psi\rangle=\frac{1}{\hbar^n}\int\overline{\varphi(\bar
z)}\psi(\bar z)e^{-\tfrac{1}{\hbar}\|z\|^2}d\mu(z). \notag
\end{equation}
The  representation~\eqref{6.8} on the space $L^2$ can be
transferred  to a representation on $\mF^2$ by means of the
Bargmann transformation
\begin{equation}
\psi(q)\longrightarrow (\mB\psi)(\bar z)\eqdef (\pi \hbar)^{-n/4}
e^{\tfrac{1}{2\hbar}\bar z^2}\int e^{-\tfrac{1}{2\hbar}(\sqrt2\,
\bar z-q)^2}\psi(q)dq,
 \label{10.1}
\end{equation}
which is a unitary isomorphism from $L^2$ onto $\mF^2$
(see~\cite{18}, \cite{19}). When realized on  $\mF^2$, the
creation operator is  the multiplication by $\bar z$ and the
annihilation operator is $\hbar\,\partial_{\bar z}$. Accordingly,
the representation of the phase-space translations is realized  by
the operators
\begin{equation}
\mB T_w\mB^{-1}=\mathcal T_w=e^{\tfrac{1}{\hbar}w\bar z-\bar
w\partial_{\bar z}}=e^{-\tfrac{1}{2\hbar}w\bar
w}e^{\tfrac{1}{\hbar}w\bar z}e^{-\bar w\partial_{\bar z}},
 \notag
\end{equation}
and  the Wick correspondence takes the form
\begin{equation}
f \longmapsto \mathfrak f=\int \tilde f(w,\bar w)\, e^{w\bar z}
e^{-\hbar\bar w\partial_{\bar z}}\,d\mu(w),
 \label{10.2}
\end{equation}

Let $1< \rho\le2$ and $A>0$. We define $\bar E^{2,A}_\rho$ to be
the space of all entire antiholomorphic functions satisfying the
condition
\begin{equation}
|\varphi(\bar z)|\le C e^{A\|z\|^2-r\|z\|^\rho},
 \label{10.3}
\end{equation}
where $C$ and $r$  are constants depending on $\varphi$,  and we
give $\bar E^{2,A}_\rho$ the topology of the inductive limit  of
Banach spaces with the norms
\begin{equation}
\|\varphi\|_r=\sup_z|\varphi(\bar z)| e^{-A\|z\|^2+r\|z\|^\rho}.
 \label{10.4}
\end{equation}

\begin{theorem}\label{T7}  Suppose $1<\rho\le 2$ and
$\rho'=\rho/(\rho-1)$. Then the Bargmann transformation defined
by~\eqref{10.1} maps the space $W^{\rho'}_\rho$ isomorphically
onto $\bar E^{2,A}_\rho$, where $A=1/(2\hbar)$.
\end{theorem}

\begin{proof}
The  transformation~\eqref{10.1} is the composition of the
following three operations: (1)  convolution by the function
$e^{-\tfrac{1}{2\hbar}\bar z^2}$, (2)  dilation by a factor of
$\sqrt2$, and (3) multiplication by  $(\pi \hbar)^{-n/4}
e^{\tfrac{1}{2\hbar}\bar z^2}$.  The first operation is equivalent
to the multiplication of $\hat \psi(\sigma)$ by
$ce^{-\tfrac{\hbar}{2}\sigma^2}$, where
$\sigma^2=\sum_j\sigma_j^2$ and the precise value of the constant
$c$ is of no importance for the proof. Let $\hat\psi\in
W^{\rho',b}_{\rho,a}(\oR^n)$ and
$\hat\psi_1(\sigma)=\hat\psi(\sigma)e^{-\tfrac{\hbar}{2}\sigma^2}$.
Then we have
\begin{equation}
|\hat\psi_1(\sigma+i\tau)|\le \|\hat\psi\|_{a,b}
e^{-\tfrac{\hbar}{2}(\sigma^2-\tau^2) -a|\sigma|^\rho+
b|\tau|^{\rho'}}. \notag
\end{equation}
Taking the inverse Laplace transform, setting   $z=x+iy$, and
using the Cauchy-Poincar\'e theorem, we get
\begin{multline}
|\psi_1(\bar z)|\le (2\pi)^{-n/2}\|\hat\psi\|_{a,b}\inf_\tau\int
\left|e^{i(x-iy)(\sigma+i\tau)}\psi_1(\sigma+i\tau)\right|d\sigma\le\\\le
C_{a'}\|\hat\psi\|_{a,b}\inf_\tau
e^{-x\tau+\tfrac{\hbar}{2}\tau^2+b|\tau|^{\rho'}} \sup_\sigma
e^{y\sigma-\tfrac{\hbar}{2}\sigma^2 -a'\|\sigma\|^\rho},
 \notag
\end{multline}
where $a'<a/n^{\rho/2}$. Because $\rho'\ge2$, we have $\tau^2\le
n+\sum_j|\tau_j|^{\rho'}$, and the infimum over $\tau$ can hence
be estimated in a manner similar to that used in proving
Theorem~\ref{T4}. Setting $B'=b+\hbar/2$, we obtain
\begin{equation}
\inf_\tau
\left(-x\tau+\frac{\hbar}{2}\tau^2+b|\tau|^{\rho'}\right)\le
\frac{n\hbar}{2}\,+\inf_\tau \left(-x\tau+
B'\sum_j|\tau_j|^{\rho'}\right)\le \frac{n\hbar}{2} -B|x|^\rho,
 \notag
\end{equation}
where $(B'\rho')^\rho(B\rho)^{\rho'}=1$. We now show that there
exists a positive constant $\epsilon$ such that
\begin{equation}
y\sigma-\frac{\hbar}{2}\sigma^2 -a'\|\sigma\|^\rho\le
\frac{1}{2\hbar}y^2 -\epsilon\|y\|^\rho\qquad\text{for}\quad
\|\sigma\|\ge\sqrt{\hbar/2}.
 \label{10.7}
\end{equation}
After a suitable rescaling of the variables, \eqref{10.7} takes
the form
\begin{equation}
(y-\sigma)^2\ge \epsilon'\|y\|^\rho -a''\|\sigma\|^\rho,
 \label{10.8}
\end{equation}
where $\epsilon'=(2\hbar)^{\rho/2}\epsilon$,
$a''=(2/\hbar)^{\rho/2} a'$, and $\|\sigma\|\ge 1$. Let us
consider the ray $y=t\sigma$, $t\ge0$ in  $\oR^{2n}$. Because
$\rho\le2$ and $\|\sigma\|\ge 1$, the inequality~\eqref{10.8}
holds  on this ray if
\begin{equation}
(t-1)^2\ge \epsilon' t^\rho -a''.
 \notag
\end{equation}
This condition in turn is satisfied for all $t$ if
$(t-1)^2\ge\epsilon' t^2-a''$, or equivalently,  $\epsilon'\le
a''/(1+a'')$.  We note that for all $\sigma$ in the ball
$\|\sigma\|\le1$, we have $(y-\sigma)^2\ge
\epsilon'\|y\|^\rho-C_{\epsilon'}$ with some constant
$C_{\epsilon'}>0$. Thus, there exists a sufficiently small
positive number $r(a,b,\hbar)$ such that
\begin{equation}
|\psi_1(x-iy)|\le
C'_{a'}\|\hat\psi\|_{a,b}e^{\tfrac{1}{2\hbar}y^2-
r\|x\|^\rho-r\|y\|^\rho}.
 \label{10.10}
\end{equation}
Therefore, the function $\varphi(\bar z)=e^{\tfrac{1}{2\hbar}\bar
z^2}\psi_1(\sqrt2\,\bar z)$ satisfies~\eqref{10.3} with
$A=1/(2\hbar)$ and with a preexponential factor proportional to
$\|\hat\psi\|_{a,b}$. We conclude that the Bargmann operator $\mB$
maps $W^{\rho'}_\rho$ continuously to $\bar E^{2, 1/2\hbar}_\rho$.

Conversely, let $\varphi\in\bar E^{2, 1/2\hbar}_\rho$ and
$\|\varphi\|_r<\infty$. Then the function $\varphi_1(\bar
z)=e^{-\tfrac{1}{4\hbar}\bar z^2}\varphi(\bar z/\sqrt2)$ satisfies
\begin{equation}
|\varphi_1(\bar z)|\le\|\varphi\|_re^{\tfrac{1}{2\hbar}y^2-
r'\|x\|^\rho-r'\|y\|^\rho},
 \notag
\end{equation}
where $r'=r/2^{1+\rho/2}$.  Hence its Fourier-Laplace transform
satisfies the estimate
\begin{multline}
|\hat\varphi_1(\sigma+i\tau)|\le (2\pi)^{-n/2}\inf_y\int
\left|\,e^{-i(x-iy)(\sigma+i\tau)}\varphi_1(x-iy)\right|dx\le\\\le
C_b\|\varphi\|_r\inf_y
e^{-y\sigma+\tfrac{1}{2\hbar}y^2-r'\|y\|^\rho}
\sup_xe^{\|x\|\cdot\|\tau\|-b\|x\|^\rho},
 \label{10.11}
\end{multline}
with $b<r'$. Substituting $y=\hbar \sigma$, we obtain
\begin{equation}
\inf_y \left(-y\sigma+\frac{1}{2\hbar}y^2-r'\|y\|^\rho\right)\le
-\frac{\hbar}{2}\sigma^2-r'\hbar^\rho\|\sigma\|^\rho.
 \notag
\end{equation}
Taking the supremum over $x$ gives the conjugate convex function
$b'\|\tau\|^{\rho'}$. Therefore, it follows from~\eqref{10.11}
that
\begin{equation}
\left|\,e^{\tfrac{\hbar}{2}(\sigma+i\tau)^2}\hat\varphi_1(\sigma+i\tau)\right|
\le C\|\varphi\|_r e^{-a|\sigma|^\rho+b''|\tau|^{\rho'}},
 \notag
\end{equation}
where $a=r'\hbar^\rho$ and $b''=b'n^{\rho'/2}$. We conclude that
$\mB^{-1}\varphi\in W^{\rho'}_\rho$ and the inverse Bargmann
transformation is continuous from $\bar E^{2,1/2\hbar}_\rho$ to
$W^{\rho'}_\rho$. The theorem is proved.
\end{proof}

\begin{corollary}\label{C1}  For any $\rho$ satisfying $1<\rho\le2$, the
Wick correspondence~\eqref{10.2} has a unique extension to a
continuous algebra homomorphism from $(\mathcal E^\rho(\oR^n),
\star_{Wick})$ to $\mathcal A_{\bar E^{2,A}_\rho}(\mF^2)$, where
$A=1/(2\hbar)$.
\end{corollary}

\section{Concluding remarks}\label{S9}

In this paper, we  consider only  analytic symbols of operators,
but the above method of extending the star products and their
representations by continuity, starting from a suitable auxiliary
space like  $\mW^\rho$ or $W_\rho^{\rho'}$, has a wider field of
application.  In particular, this approach allows us to extend the
Weyl symbol calculus  beyond the traditional framework of tempered
distributions, see~\cite{20}, \cite{21}. Such a generalization is
also desirable for the anti-Wick correspondence because not all
bounded operators have anti-Wick symbols among the usual
functions~\cite{18}.

If $\rho>2$, then the  space $W_\rho^{\rho'}$ is not an algebra
under the Wick star product. However it is a topological algebra
with respect to the Weyl-Moyal product $\star_\hbar$. It follows
from a theorem of~\cite{17} that the elements of $\mathcal E^\rho$
are two-sided multipliers of the algebra
$(S^{1/\rho}_{1/\rho},\star_\hbar)$ which coincides with the
algebra $(W_\rho^{\rho'},\star_\hbar)$ for  $\rho>1$. This implies
that for all $f\in \mathcal E^\rho$ and $\varphi\in
W_\rho^{\rho'}$, the products $f\star_\hbar\varphi$ and
$\varphi\star_\hbar f$ belong to $W_\rho^{\rho'}$ and are
continuous in $f$ and $\varphi$. Furthermore, the following
associativity relations hold:
\begin{equation}
(f\star_\hbar\varphi)\star_\hbar\psi=f\star_\hbar(\varphi\star_\hbar\psi),
\quad (\varphi\star_\hbar
f)\star_\hbar\psi=\varphi\star_\hbar(f\star_\hbar\psi), \quad
(\varphi\star_\hbar \psi)\star_\hbar
f=\varphi\star_\hbar(\psi\star_\hbar f).
  \notag
 \end{equation}
As shown  in~\cite{21}, the Weyl correspondence defines an
isomorphism between the left multiplier algebra of
$(W_\rho^{\rho'},\star_\hbar)$ and the operator algebra $\mathcal
A_{W^{\rho'}_\rho}(L^2)$. Therefore, for the case of Weyl-Moyal
product, Theorem~\ref{T6} can be derived as a consequence of these
results. The direct proof given here is more illuminating and
applies to any star product of the form~\eqref{2.2}.

Finally, we note that along  with the spaces $\mathcal E^\rho$, it
is natural to consider the spaces of entire functions of the same
order and finite type, i.e., satisfying the condition
\begin{equation}
|f(z)|\le Ce^{a|z|^\rho},
  \notag
 \end{equation}
where $a$ and $C$  are constants depending on $f$. We let $E^\rho$
denote this space and  endow it with the inductive limit topology
defined by the family of Banach spaces $E^{\rho, a}$ with the
norms $\|f\|_a=\sup_z|f(z)|e^{-a|z|^\rho}$.  It is easy to see
that an analogue of Theorem~\ref{T2} holds for $E^\rho$, but only
if the strong inequality $\rho<2$ holds. It is essential that if
$1/\rho'+1/\rho=1$ and $\rho\le2$, then the series defining the
star product $f\star_\vartheta g$ converges absolutely in
$E^{\rho'}$ for all $f\in\mathcal E^\rho$ and $g\in E^{\rho'}$.
Moreover, it can be proved that the algebra $(\mathcal E^\rho,
\star_\vartheta)$ acts continuously on the space $E^{\rho'}$ and
on its subspace $W^{\rho'}_\rho$ by the left and right
$\star_\vartheta$-multiplication. In other words, the function
spaces  $E^{\rho'}$ and $W^{\rho'}_\rho$ (as well as the analogous
spaces associated with cone-shaped regions) are topological
bimodules over the algebra $(\mathcal E^\rho, \star_\vartheta)$
for any $\vartheta$.

\begin{acknowledgements}  This paper was supported in part by the the
Russian Foundation for Basic Research (Grant No.~12-01-00865)
\end{acknowledgements}


\begin{thebibliography}{49}

\bibitem{1} H.~Omori, Y.~Maeda, N.~Miyazaki, and  A.~Yoshioka,
``Deformation quantization of Fr\'echet-Poisson algebras --
Convergence of the Moyal product,'' in: {\it Conference Moshe
Flato 1999. Quantization, Deformation, and Symmetries,} edited by
G.~Dito, D.~Sternheimer, Math. Physics Studies, no.~ 22,
(Dordrecht, Kluwer, 2000), Vol.~2, pp.~233-245.


\bibitem{2} H.~Omori, Y.~Maeda, N.~Miyazaki, and A.~Yoshioka,
``Non-formal deformation quantization of Fr\'echet-Poisson
algebras: The Heisenberg Lie algebra case,''  Contemp. Math. {\bf
434} (2007) 99-124.


\bibitem{3} S.~Beiser, H.~R\"omer, and S.~Waldmann,
``Convergence of the Wick product,''  Commun. Math. Phys. {\bf
272} (2007) 25-52 [arXiv:math.QA/0506605].


\bibitem{4} S.~Beiser and S.~Waldmann, ``Fr\'echet algebraic deformation
quantization of the Poincar\'e disk,'' e-print arXiv:1108.2004.


\bibitem{5} O.~W.~Greenberg,  ``Failure of microcausality in quantum
field theory on non\-commutative spacetime,''  Phys.Rev. D {\bf
73} (2006)  045014  [arXiv:hep-th/0508057].


\bibitem{6}  M.~A.~Soloviev,  ``Noncommutativity and
$\theta$-locality,''  J. Phys A: Math. Theor. {\bf 40} (2007)
14593-14604 [arXiv:0708.1151].

\bibitem{7} M.~Chaichian, M.~N.~Mnatsakanova,  A.~Tureanu,
and Yu.~A.~Vernov, ``Test function space in  noncommutative
quantum field theory,''   JHEP {\bf 0809}  (2008) 125
[arXiv:0706.1712].


\bibitem{8}  M.~A.~Soloviev, ``Failure of microcausality in
noncommutative field theories,''   Phys. Rev.  D {\bf 77} (2008)
125013 [arXiv:0802.0997]


\bibitem{9} S.~Galluccio, F.~Lizzi, and P.~Vitale, ''Twisted noncommutative
field theory with the Wick-Voros and Moyal products,''  Phys. Rev.
D {\bf 78} (2008) 085007 [arXiv:0810.2095].


\bibitem{10} A.~P.~Balachandran and M.~Martone, ''Twisted quantum fields on
Moyal and Wick-Voros planes are inequivalent,''  Mod. Phys. Lett.
A {\bf 24} (2009) 1721-1730 [arXiv:0902.1247].


\bibitem{11} A.~P.~Balachandran, A.~Ibort, G.~Marmo, and M.~Martone,
``Inequivalence of QFT's on noncommutative spacetimes: Moyal
versus Wick-Voros,''  Phys. Rev. D {\bf 81} (2010) 085017
[arXiv:0910.4779].


\bibitem{12} M.~A.~Soloviev, ``Noncommutative deformations of
quantum field theories, locality, and causality,''  Theor. Math.
Phys. {\bf 163} (2010) 741-752 [arXiv:1012.3536].



\bibitem{13} M.~Blaszak and Z. Doma\'nski,  ``Phase space  quantum
mechanics,''  Ann. of Phys. {\bf 327} (2012) 167-211
[arXiv:1009.0150].



\bibitem{14}  H.~Omori, Y.~Maeda, N.~Miyazaki, and A.~Yoshioka,
``Deformation expression for elements of algebras (II) –(Weyl
algebra of 2m-generators),'' e-print arXiv:1105.1218.



\bibitem{15}   I.~M.~Gelfand and G.~E.~Shilov, {\it Generalized
Functions}, Vol.~3 (Academic, New York 1968).


\bibitem{16} J.~M.~Gracia-Bondia and J.~C.~V\'arilly, ``Algebras
of distributions suitable for phase-space quantum mechanics. I.''
 J. Math. Phys. {\bf 29} (1988)  869-879.


\bibitem{17}   M.~A.~Soloviev, ``Twisted convolution and Moyal star
product of generalized functions,''   Theor. Math. Phys. {\bf 172}
(2012) 885-900 [arXiv:1208.1838].


\bibitem{18} G.~B.~Folland, {\it Harmonic Analysis in Phase Space},
Annals of Math. Studies, Vol. 122 (Princeton Univ. Press,
Princeton, 1989).


\bibitem{19} F.~A.~Berezin and M.~A.~Shubin, {\it Schr\"odinger
Equation}  (Dordrecht, Kluwer, 1991).


\bibitem{20} M.~A.~Soloviev, ``Moyal multiplier algebras of the test function spaces
of type S,''  J. Math. Phys.  {\bf 52} (2011) 063502
[arXiv:1012.0669].

\bibitem{21} M.~A.~Soloviev, ``Generalized Weyl correspondence and Moyal
multiplier algebras,'' Theor. Math. Phys. {\bf 173} (2012)
1359-1376.



\end{thebibliography}
\end{document}